\documentclass[10pt, final, journal, letterpaper, twocolumn]{IEEEtran}

\usepackage{cite}
\usepackage{amsmath,amsthm,amssymb,amsfonts}
\usepackage{algorithm}
\usepackage{algorithmic}
\usepackage{graphicx}
\usepackage{color}
\usepackage[normalsize]{subfigure}

\newtheorem{proposition}{Proposition}

\begin{document}

\title{Joint Resource Allocation in SWIPT-Based Multi-Antenna Decode-and-Forward Relay Networks}
\author{Yuan~Liu,~\IEEEmembership{Member,~IEEE}
\thanks{This paper was presented in part at the IEEE International Conference on Communications, Paris, France, May 2017 \cite{icc17conf}. }
\thanks{The author is with the School of Electronic and Information Engineering, South China University of Technology, Guangzhou 510641, China (e-mail: eeyliu@scut.edu.cn).}
}

\maketitle

\vspace{-1.5cm}

\begin{abstract}
In this paper, we consider relay-assisted simultaneous wireless information and power transfer (SWIPT) for two-hop cooperative transmission, where a half-duplex multi-antenna relay adopts decode-and-forward (DF) relaying strategy for information forwarding. The relay is assumed to be energy-free and needs to harvest energy from the source node.  By embedding power splitting (PS) at each relay antenna to coordinate the received energy and information,  joint problem of determining PS ratios and power allocation at the multi-antenna relay node is formulated to maximize the end-to-end achievable rate. We show that the multi-antenna relay is equivalent to a virtual single-antenna relay in such a SWIPT system, and the problem is optimally solved with closed-form. To reduce the hardware cost of the PS scheme, we further propose the antenna clustering scheme, where the multiple antennas at the relay are partitioned into two disjoint groups which are exclusively used for information decoding and energy harvesting, respectively. Optimal clustering algorithm is first proposed but with exponential complexity. Then a greedy clustering algorithms is introduced with linear complexity and approaching to the optimal performance. Several valuable insights are provided via theoretical analysis and simulation results.
\end{abstract}


\section{Introduction}

Energy harvesting is a promising method to prolong the lifetime of energy-constrained wireless networks. Compared with conventional energy supplies such as batteries with fixed operation time, energy harvesting from surrounding environment potentially provides an immortal energy supply. However, the conventional  energy harvesting depends on natural energy sources (like solar, wind, vibration and so on), which can not be controlled and are not always available. Recently, radio-frequency (RF) signals radiated by transmitters provide self-sustainable and controllable energy source for wireless energy harvesting and thus attracts considerable research interests \cite{LuCST,ChenCM}. Since RF signals carry both energy and information at the same time, simultaneous wireless information and power transfer (SWIPT) processing enables new resource allocation schemes at transceivers and thus has drawn a significant attention in wireless communications.

The prior works  \cite{Varshney,Grover} investigated the fundamental ``energy-rate" tradeoff of  SWIPT, where however the receiver is able to perform information decoding and energy harvesting independently from the same received signal without any loss, which is not practically realizable yet due to the current circuit limitations as pointed out by \cite{ZhangHo,ZhouZhang}. Thus the authors of \cite{ZhangHo,ZhouZhang}  proposed two practical designs with separated information decoding
and energy harvesting receiver for SWIPT, namely \emph{``time switching"} (TS) and \emph{``power splitting" }(PS). If the TS is  employed at the receiver, the received signal is  processed by either energy harvesting or information decoding. With PS
employed at the receiver, the received signal is split into two signal streams with a certain ratio by a power splitter, where one stream is to the energy receiver and the other one is to the information receiver. The authors in \cite{LiuZhang1,LiuZhang2}  investigated the tradeoff between the ergodic rate and average energy for SWIPT. In \cite{Shi}, the transmit beamforming  and the PS strategy were jointly optimized for the multiple-input single-output (MISO) multiuser system. The authors in \cite{Ng2} studied SWIPT based energy-efficiency in downlink orthogonal frequency division multiplexing (OFDM) systems. SWIPT has been also considered as an efficient solution for physical layer security \cite{Ng1,ZhangTIFS,MengTWC,MLiu}.

In relay or sensor networks, the intermediate relay (or sensor) nodes often have limited battery storage and require external charging to remain active. Therefore, SWIPT is more important and applicable in relay or sensor networks. Several works investigated SWIPT in cooperative single-input single-output (SISO) relay systems. For instance, outage and ergodic capacity for both TS and PS were derived in \cite{Nasir}, where the relay has the energy harvesting function and harvests a fraction of energy from the source, then the relay uses the harvested energy to forward the source's information to the destination. This is referred to as the \emph{``harvest-then-use"} energy harvesting system. Both decode-and-forward (DF) and amplify-and-forward (AF)  were studied in this work. The optimization of energy arrivals for throughput maximization using Lagrangian duality was proposed in \cite{Gurakan}  for multiuser full-duplex relay system.
Outage probability and diversity gain of SWIPT were characterized in \cite{Ding} where multiple relays assist a source-destination pair.  Power allocation strategies were studied in \cite{Ding2}  where a relay node with energy harvesting function assists multiple source-destination transmissions. The authors in \cite{Zhong} investigated the optimal TS ratio for full-duplex relaying systems. In \cite{ChenGame2015}, game theory for interference relay channels with PS was studied.  The authors in \cite{Gong} studied the TS  in a multi-relay network with relay beamforming. Outage of PS-based SWIPT with DF full-duplex was studied in \cite{LiuTW16}.
The author in \cite{YuanCL} studied optimal PS strategies for both AF and DF relaying in multi-relay assisted cooperative networks. Two protocols with/without direct link and optimal resource allocation schemes for SWIPT based OFDM relay system were investigated in \cite{YuanVT}.

Since multiple-input multiple-output (MIMO) or multi-antenna technique has been adopted as an efficient solution to achieve high spectral efficiency for current and future broadband wireless systems, a handful of recent works also discussed  SWIPT in MIMO relay systems.
Efficient algorithms for SWIPT in MIMO-OFDM AF relay system were proposed in \cite{Xiong2015}, where the relay harvests and then uses the energy from the source for information forwarding.
Self-energy recycling for wireless-powered full-duplex MIMO AF relay was studied in \cite{Zeng1}, in which one of the relay antennas is used to harvest energy from the source and the others are used to receive information. The authors in \cite{Quanzhong2014} studied SWIPT in  MIMO AF relay systems, where secure relay beamforming was designed for a destination, eavesdropper and energy receiver.
Antenna clustering methods were proposed in \cite{Krikidis,ZZhou}, where the multiple antennas of the relay node are partitioned into two disjoint sets, with one for energy harvesting and the other for information decoding. SWIPT in  MIMO AF was also investigated in \cite{Huang1,Huang2}, where the relay harvests energy from the source's information flow and the destination's energy flow. In \cite{Mohammadi}, the authors considered TS in full-duplex MIMO DF relaying, where time allocation with different precoder designs were proposed.

In view of these related works on  SWIPT based multi-antenna relay systems, it is found that most works focused on AF relaying strategy, and
the DF relaying strategy is much less investigated. Though \cite{Mohammadi} considered DF relaying, the work only considered TS at the relay node. Thus, the optimal transceiver architecture design and optimal wireless resource allocation for SWIPT based multi-antenna DF relaying leave a large space to be exploited.
This motivates our paper.

In this paper, we consider a classical three-node cooperative relay transmission, where the half-duplex relay node is equipped with multiple antennas and the source and destination are equipped with a single antenna. The source node is with fixed energy supply, and the relay node has no energy or is not willing to expend its own energy to help the source, i.e., the harvest-then-use based SWIPT is adopted. Specifically,  assume that the relay node has the energy harvesting function, it harvests the energy from the source's signal for helping the information transmission from the source to the destination. The main contributions of this paper are summarized as follows:

\begin{itemize}
\item We investigate PS based SWIPT in multi-antenna  relay system using DF relaying strategy, where power splitting ratios and the allocation of the harvested power at the relay node are jointly optimized to maximize the end-to-end information rate. The distinct feature of the formulated problem is that the multiple relay antennas have individual power splitters to make the system more flexible. By doing so, an additional dimension of spatial diversity can be explored and hence the new resource allocation problem becomes complicated and challenging.
\item Efficient algorithm is proposed to optimally solve the joint optimization problem. In particular, we reveal that the power splitting ratios of the relay antennas should be identical at the optimum. Moreover, we show that the multi-antenna relay system is equivalent to a \emph{virtual} single-antenna relay system.
\item To ease the hardware implementation of the PS scheme, we propose the antenna clustering scheme, where the multiple antennas of the relay are partitioned as two disjoint groups with one for information decoding and the other for energy harvesting. Optimal and suboptimal clustering methods are proposed. It is notable that the proposed suboptimal method is only with linear complexity and approaches to the optimal performance.
\item Valuable insights are provided via simulations. In particular, we show that the optimal power splitting ratio of PS remains unchanged with the transmit power and only depends on the channel conditions of the second hop, while the optimal time allocation factor of TS is decreasing with the transmit power.
\end{itemize}

The rest of this paper is  organized as follows. Section II describes the system model. Section III and Section IV present the problem formulation and optimal solution of the PS scheme. Section V proposes the antenna clustering problem and the corresponding solutions. Simulation results and discussions are provided in Section VI. Finally, Section VII concludes this paper.

\section{System Model}

\begin{figure}[t]
\begin{centering}
\includegraphics[scale=0.8]{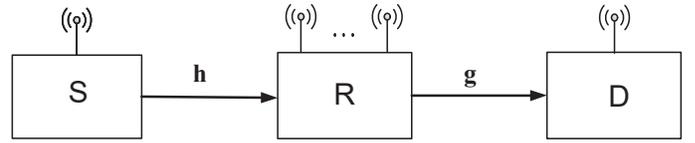}
\vspace{-0.1cm}
 \caption{System model.}\label{fig:system}
\end{centering}
\vspace{-0.3cm}
\end{figure}

We consider a two-hop cooperative relay network as shown in Fig. \ref{fig:system}, where the relay node $\textsf{R}$ is equipped with multiple antennas and the source $\textsf{S}$ and destination $\textsf{D}$ are equipped with a single antenna.  The considered three-node configuration is a very general model that can be applied to many wireless communication applications, such as cellular or ad hoc network. If in a cellular network, the communication of a D2D pair (i.e., the single-antenna source-destination pair) can be assisted by a multi-antenna relay node. If in an ad hoc network, a multi-antenna transmitter without its own transmission task at some time can help another pair (i.e., the single-antenna source-destination pair). 
The relay node is half-duplex for practical consideration. The antenna set of the relay node is denoted as $\mathcal A=\{1,\cdots,N\}$. It is assumed that the direct link between the source and the destination is unavailable due to the shielding effect caused by obstacles. This is the well known Type-II relay model in the 3rd generation partnership project long term evolution advanced (3GPP LTE-A).
We assume that the additive white Gaussian noises (AWGN) at all nodes are independent circular symmetric complex Gaussian random variables, each having zero mean and unit variance.  The transmission from the source to destination is divided into consecutive frames, where the channel fading remains unchanged within each transmission frame but varies from one frame to another. We also assume that perfect channel state information (CSI) are available for centralized processing, which can be obtained for the following way: In the training phase, the relay broadcasts pilot signal, while the source and destination receive the pilot signal for channel estimation and then feed the CSI back to the relay. The relay collects the global CSI for resource allocation.
%


The relay node has the energy harvesting function to harvest energy from the received signals by PS  employed at the relay receiver.  The relay has no energy (or does not freely expend its own energy) to help the source, but it can forward the source's information by using the energy harvested from the source.

There are two points about the multi-antenna DF relaying to be noted. First, the received signal at every relay antenna can be jointly decoded at the relay. Second, the harvested power at all relay antennas in the first hop are added up as a total power for information forwarding over all relay antennas in the next hop.

\begin{figure}[t]
\begin{centering}
\includegraphics[scale=1]{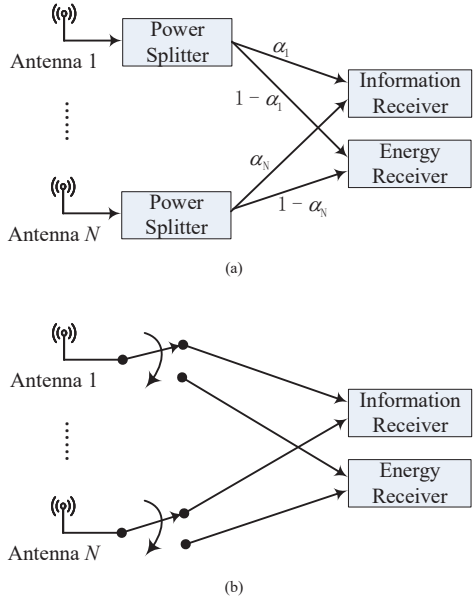}
\vspace{-0.1cm}
 \caption{Illustration of multi-antenna relay receiver: (a) power splitting scheme and (b) antenna clustering scheme.}\label{fig:protocol}
\end{centering}
\vspace{-0.3cm}
\end{figure}

\section{ Problem Formulation and Optimal Solution for the PS Protocol}

In this section, the joint optimization problem of determining PS ratios and power allocation at the multi-antenna relay is studied.

\subsection{Problem Formulation}
We use $\tilde{ h}_i$ to denote the complex channel coefficient of the source to relay antenna $i$, $\tilde{g}_i$ the complex channel coefficient of relay antenna $i$ to the destination, and denote $h_i=|\tilde{ h}_i|^2$ and $g_i=|\tilde{g}_i|^2$ as the channel power gains.

For the PS embedded multi-antenna relay, the transmission frame is divided into two equal phases. At the first phase, the source transmits signal $x$ to the relay node, and the received RF signal at relay antenna $i$ is $r_i=\tilde{h}_ix+n_i$, where $E[|x|^2]=P$ is the source's power and $n_i$ is the received AWGN at relay antenna $i$. Denote $0\leq\alpha_i\leq1$ as the power splitting ratio at relay antenna $i$, as shown in Fig. \ref{fig:protocol}(a). By ignoring the received AWGN $n_i$, i.e., $r_i=\tilde{h}_ix$, the received RF signal used to information decoding is $\sqrt{\alpha_i}r_i=\sqrt{\alpha_i}\tilde{h}_ix$ whose power is  $Ph_i$, and the received RF signal used to energy harvesting is $\sqrt{1-\alpha_i}r_i=\sqrt{1-\alpha_i}\tilde{h}_ix$ whose power is  $Ph_i(1-\alpha_i)$. Note that at the information receiver of the relay node, the received RF signal $\sqrt{\alpha_i}r_i$ is first converted to a complex baseband signal and then digitalized by an analog-to-digital converter (ADC) for further decoding \cite{ZhouZhang}. Thus the ADC output for decoding is $r_i'=r_i+n_{{\rm cov}}=\sqrt{\alpha_i}\tilde{h}_ix+n_{{\rm cov}}$, where $n_{{\rm cov}}$ is the noise introduced by the RF
band to baseband signal conversion and assumed to be zero mean and unit variance.

Let $p_i$ denote the transmit power of relay antenna $i$ in the second hop, we can easily obtain the end-to-end achievable rate of the PS as
\begin{eqnarray}\label{eqn:Rps}
  R_{PS}\leq\frac{1}{2}\min\Bigg\{\log_2\left(1+\sum_{i\in\mathcal A} Ph_i\alpha_i\right),\nonumber\\
  \log_2\left(1+\left(\sum_{i\in\mathcal A} \sqrt{p_ig_i}\right)^2\right)\Bigg\}.
\end{eqnarray}
Here the pre-log factor $\frac{1}{2}$ is due to the fact that two phases are used for information transmission. Note that the first term of the min-operator in \eqref{eqn:Rps} is because of that the information received at the multiple antennas of the relay can be decoded jointly in the first hop, and the second term is due to the fact that the second hop is actually a MISO channel \cite[p.179]{Tse}.

Our goal is to maximize the end-to-end achievable rate by jointly determining the transmit power and
power splitting ratio on each relay antenna. Let $\boldsymbol\alpha=[\alpha_1,\cdots,\alpha_N]^T$ and $\mathbf p=[p_1,\cdots,p_N]^T$, the problem can be mathematically formulated as
\begin{subequations}
\begin{align}
\textbf{P1:}~~~  \max_{\boldsymbol\alpha,\mathbf p, R_{PS}}~~&R_{PS}\\
  s.t.~~~~&\sum_{i\in\mathcal A}p_i\leq\sum_{i\in\mathcal A}\zeta Ph_i(1-\alpha_i)\label{eqn:rpcons1}\\
  &0\leq\alpha_i\leq1,~\forall i\in\mathcal A\label{eqn:rpcons2}\\
  &p_i\geq0,~\forall i\in\mathcal A,
\end{align}
\end{subequations}
where $\zeta$ is the energy conversion efficiency, and the constraint \eqref{eqn:rpcons1} ensures that the transmit power of the relay node can not exceed its harvested power. Since the circuit power consumption can be assumed as a constant in general, adding the the circuit power consumption in the power constraint \eqref{eqn:rpcons1} does not affect our algorithm in the sequel.


\subsection{Optimal Solution}

Note that  the region of $R_{PS}$ is a  convex set, the constraint \eqref{eqn:rpcons1} is linear, and the constraint \eqref{eqn:rpcons2} is affine. Therefore, the problem \textbf{P1} is a convex problem, and we can use the Lagrangian dual method to find the globally optimal solution.

We first let the non-negative Lagrangian multipliers $\lambda$ and $\beta$ associate with the two rate constraints of $R_{PS}$, and $\mu$ with the total power constraint \eqref{eqn:rpcons1}. The Lagrangian of \textbf{P1} is written as
\begin{align}\label{eqn:Laps}
  &L(\lambda,\beta,\mu,\boldsymbol\alpha,\mathbf p,R_{PS})\nonumber\\
  &=R_{PS}+\lambda\left[\frac{1}{2}\log_2\left(1+\sum_{i\in\mathcal A} Ph_i\alpha_i\right)-R_{PS}\right]\nonumber\\
  &+\beta\left[\frac{1}{2}\log_2\left(1+\left(\sum_{i\in\mathcal A} \sqrt{p_ig_i}\right)^2\right)-R_{PS}\right]\nonumber\\
  &+\mu\left[\sum_{i\in\mathcal A}\zeta Ph_i(1-\alpha_i)-\sum_{i\in\mathcal A}p_i\right].
\end{align}

Denote $\mathcal D$ as the set of $\{\boldsymbol\alpha,\mathbf p,R_{PS}\}$ satisfying the primary constraints, then the dual function of \textbf{P1} is given by
\begin{equation}
g_{PS}(\lambda,\beta,\mu)=\max_{\{\boldsymbol\alpha,\mathbf p,R_{PS}\}\in\mathcal D}L(\lambda,\beta,\mu,\boldsymbol\alpha,\mathbf p,R_{PS}).
\end{equation}

To compute the dual function $g_{PS}(\lambda,\beta,\mu)$, we need to find the optimal $\{\boldsymbol\alpha^*,\mathbf p^*\}$ to maximize the Lagrangian under the given dual variables $\{\lambda,\beta,\mu\}$.

The part of the dual function with respect to the rate variable $R_{PS}$ can be expressed as
\begin{equation}
g_0(\lambda,\beta)=\max_{R_{PS}\geq0}(1-\lambda-\beta)R_{PS}.
\end{equation}
To make sure that the dual function is bounded, the condition $(1-\lambda-\beta)=0$ must hold such that $g_0(\lambda,\beta)\equiv0$ \cite{YuanTWC13}, which implies that $\beta=1-\lambda$. Note that $0\leq\lambda\leq1$ such that $\beta$ is non-negative. Then we can remove the variables $R_{PS}$ and $\beta$ in following derivations.

By removing $R_{PS}$ and $\beta$ in \eqref{eqn:Laps}, the Lagrangian is rewritten as
\begin{align}
  L(\lambda,\mu,\boldsymbol\alpha,\mathbf p)=&\frac{\lambda}{2}\log_2\left(1+\sum_{i\in\mathcal A}Ph_i\alpha_i\right)\nonumber\\
  &+\frac{1-\lambda}{2}\log_2\left(1+\left(\sum_{i\in\mathcal A} \sqrt{p_ig_i}\right)^2\right)\nonumber\\
  &-\mu\sum_{i\in\mathcal A}p_i-\mu\sum_{i\in\mathcal A}\zeta Ph_i\alpha_i.
\end{align}

We first derive the optimal power splitting ratios $\boldsymbol\alpha^*$ on the relay antennas in the following proposition.

\begin{proposition}\label{prop:a}
For the multi-antenna DF relaying, the optimal power splitting ratios on all relay antennas should be the same. Define $\alpha_1=\cdots=\alpha_N\triangleq\alpha$, the optimal identical power splitting ratio $\alpha^*$ is
\begin{equation}\label{eqn:opta}
  \alpha^*=\left[\frac{1}{\sum_{j\in\mathcal A}Ph_j}\left(\frac{\lambda^*}{\delta\mu^*\zeta}-1\right)\right]_0^1,
\end{equation}
where $\delta\triangleq2\ln2$ and $[x]_a^b\triangleq\max\{\min\{x,b\},a\}$.
\end{proposition}

\begin{proof} For any given dual variables $\{\lambda,\mu\}$, it is readily to verify that the Lagrangian $L(\lambda,\mu,\boldsymbol\alpha,\mathbf p)$ is  concave in $\boldsymbol\alpha$. By applying the optimality Karush-Kuhn-Tucker (KKT) conditions \cite{Boyd} with respect to $\boldsymbol\alpha$, the following conditions always hold at the optimal dual point $\{\lambda^*,\mu^*\}$:

\begin{eqnarray}
  \frac{\partial L(\lambda,\mu,\boldsymbol\alpha,\mathbf p)}{\partial\alpha_i^*}\begin{cases}<0,~\alpha_i^*=0\\
  =0,~0<\alpha_i^*<1\\
  >0,~\alpha_i^*=1\end{cases}\forall i,
\end{eqnarray}
where
\begin{align}
\frac{\partial L(\lambda,\mu,\boldsymbol\alpha,\mathbf p)}{\partial\alpha_i^*}&=\frac{\lambda^* Ph_i}{\delta(1+\sum_{j\in\mathcal A}Ph_j\alpha_j^*)}-\mu^*\zeta Ph_i\nonumber\\
&=Ph_i\left[\frac{\lambda^* }{\delta(1+\sum_{j\in\mathcal A}Ph_j\alpha_j^*)}-\mu^*\zeta \right].
\end{align}

As the term $\frac{\lambda^*}{\delta(1+\sum_{j\in\mathcal A}Ph_j\alpha_j^*)}-\mu^*\zeta$ in above is identical for all $i$ and thus is a constant, it must be $\alpha_1^*=\cdots=\alpha_N^*$. Substituting the result to this constant term and equating it to be zero, \eqref{eqn:opta} can be obtained.

This completes the proof.
\end{proof}

We then turn to the optimal power allocations $\mathbf p^*$ of the relay node.
Denote $P_R\triangleq\sum_{i\in\mathcal A}\zeta Ph_i(1-\alpha_i^*)$ as the total harvested power of the relay for given power splitting ratio $\alpha_i^*=\alpha^*$. As the second hop is actually a MISO channel, the optimal power allocation of the relay node follows the maximal-ratio combining (MRC) (see Appendix \ref{app:p}), i.e.,
\begin{equation}\label{eqn:optp}
p_i^*= \frac{g_i}{\sum_{j\in\mathcal A}g_j}P_R,~\forall i.
\end{equation}

%

After finding the optimal $\{\boldsymbol\alpha^*,\mathbf p^*\}$, we turn to solve the dual problem which can be expressed as
\begin{align}
\min_{\{\lambda,\mu\}}~&g_{PS}(\lambda,\mu)\nonumber\\
s.t.~~&0\leq\lambda\leq1,\mu\geq0.
\end{align}

As a dual function is always convex \cite{Boyd}, we adopt the ellipsoid method to simultaneously iterate the dual variables $\lambda$ and $\mu$ to the optimal ones by using the defined subgradients as follows
\begin{align}\label{eqn:sg1}
\Delta\lambda=&\frac{1}{2}\log_2\left(1+\sum_{i\in\mathcal A} Ph_i\alpha_i^*\right)\nonumber\\
  &-\frac{1}{2}\log_2\left(1+\left(\sum_{i\in\mathcal A} \sqrt{p_i^*g_i}\right)^2\right),
\end{align}

\begin{equation}\label{eqn:sg2}
\Delta\mu=\sum_{i\in\mathcal A}\zeta Ph_i(1-\alpha_i^*)-\sum_{i\in\mathcal A}p_i^*.
\end{equation}

So far we have solved \textbf{P1} optimally by the routine of the dual method, which is an iterative algorithm (see Algorithm 1 formally) for updating the dual variables. We can also find the optimal solution for \textbf{P1} without any iteration by carefully exploring Proposition \ref{prop:a} and \eqref{eqn:optp} in the dual method.

\subsection{Equivalence to Single-Antenna Relay Case}

To obtain more insights, we first consider a special case where the relay node is equipped with one antenna. In this case, assume that $h$ and $g$ are the channel gains of the first and second hops, \textbf{P1} becomes
\begin{align}\label{eqn:sa}
\max_{0\leq\alpha\leq1} R_{PS}\leq\frac{1}{2}\min\{\log_2\left(1+ Ph\alpha\right),\nonumber\\
\log_2\left(1+ \zeta Ph(1-\alpha)g\right)\}.
\end{align}
It is straightforward that the optimal solution of the problem \eqref{eqn:sa} must happen at $\log_2\left(1+ Ph\alpha\right)=\log_2\left(1+ \zeta Ph(1-\alpha)g\right)$, which results in $\alpha^*=\frac{\zeta g}{1+\zeta g}$.
%
%
It is interesting that $\alpha^*$ in above only depends on the energy conversion efficiency and the second hop channel gain, and is regardless of the source's transmit power and the channel gain of the first hop.

Then, based on Proposition \ref{prop:a} and \eqref{eqn:optp} obtained in the dual method, and the above single antenna example, we have the following proposition:

\begin{proposition}\label{prop:sa}
The multi-antenna relay can be regarded as a ``virtual" single-antenna relay by letting $h\triangleq\sum_{i\in\mathcal A}h_i$ and $g\triangleq\sum_{i\in\mathcal A}g_i$. Thus \textbf{P1} is equivalent to the problem \eqref{eqn:sa} and the optimal solution is

\begin{align}
\alpha^*&=\frac{\zeta g}{1+\zeta g},
\end{align}
and the optimal achievable rate is

\begin{equation}
R_{PS}^* = \frac{1}{2}\log_2\left(1+\frac{\zeta Phg}{1+\zeta g}\right).
\end{equation}
\end{proposition}

\begin{proof}
First, based on Proposition \ref{prop:a} that all relay antennas have the same power splitting ratio $\alpha$, we can write the information rate of the first hop as $\frac{1}{2}\log_2(1+Ph\alpha)$, and the harvested power is $P_R=\zeta Ph(1-\alpha)$, where $h\triangleq\sum_{i\in\mathcal A}h_i$. Based on \eqref{eqn:optp} that all harvested power $\zeta Ph(1-\alpha)$ is proportionally allocated to each relay antenna based on its channel gain of the second hop, we obtain that the received signal amplitude at the destination from each relay antenna $i$ is $\sqrt{p_ig_i}=g_i\sqrt{\frac{P_R}{\sum_{j\in\mathcal A}g_j}}$, and thus the SNR at the destination is $(\sum_{i\in\mathcal A}\sqrt{p_ig_i})^2=P_R\sum_{i\in\mathcal A}g_i$. Let $g\triangleq \sum_{i\in\mathcal A}g_i$, the information rate of the second hop is $\frac{1}{2}\log_2(1+P_Rg)$.  Therefore, $R_{PS}$ defined in \eqref{eqn:Rps} is equivalent to the objective function of the single-antenna problem \eqref{eqn:sa}. Applying the similar method as for the problem \eqref{eqn:sa}, the conclusions are obtained.
\end{proof}

Proposition \ref{prop:sa} establishes the equivalence between the optimal PS for the multi-antenna and single-antenna DF relaying systems. It also suggests that we can treat the multiple relay antennas as a ``virtual" single antenna with the channel gains of the two hops as $h=\sum_{i\in\mathcal A}h_i$ and $g=\sum_{i\in\mathcal A}g_i$, respectively.

%
%
%
\begin{algorithm}[!t]
\caption{Dual Method for \textbf{P1}}
\begin{algorithmic}[1]
\STATE \textbf{initialize} $\{\lambda,\mu\}$ as non-negative values.
\REPEAT \STATE Find the optimal power splitting ratios $\boldsymbol\alpha^*(\lambda,\mu)$ using Proposition \ref{prop:a}. \STATE Compute the optimal power allocations $\mathbf p^*(\lambda,\mu)$ using  \eqref{eqn:optp}.  \STATE Update $\{\lambda,\mu\}$ by the ellipsoid method using the subgradients defined in \eqref{eqn:sg1} and \eqref{eqn:sg2}.
 \UNTIL{$\{\lambda,\mu\}$ converge.}
\end{algorithmic}
\end{algorithm}


\section{Antenna Clustering Scheme}

In Section III, we propose the optimal solution for the PS, in which a power splitter is required on each relay antenna to adjust the power splitting ratio. However, this could be very costly to implement in practice. Thus in this section, we introduce an antenna clustering scheme as shown in Fig. \ref{fig:protocol}(b). In this scheme, instead of splitting the power at each relay antenna, the relay antenna set $\mathcal A$ is divided into two disjoint subsets $\Omega_I$ and $\Omega_E$, where the relay antennas in $\Omega_I$ are exclusively used for information decoding and the others in $\Omega_E$ are exclusively used for energy harvesting. That is,

\begin{equation}
\alpha_i=\begin{cases}1,~{\rm if}~ i\in\Omega_I\\0,~{\rm if}~ i\in\Omega_E\end{cases}\forall i.
\end{equation}

Compared with the PS, the antenna clustering scheme is practically more favorable since the antenna clustering only needs time switcher at each relay antenna and the power splitters of the PS are more costly.

Then the end-to-end achievable rate of the antenna clustering scheme is
\begin{eqnarray}
  R_{AC}\leq\frac{1}{2}\min\Bigg\{\log_2\left(1+\sum_{i\in\Omega_I} Ph_i\right), \nonumber\\
  \log_2\left(1+\left(\sum_{i\in\mathcal A}\sqrt{p_ig_i}\right)^2\right)\Bigg\}.
\end{eqnarray}

The problem is to maximize the end-to-end rate by jointly partitioning the relay antenna set and allocating the harvested power, which can formulated as
\begin{subequations}
\begin{align}
\textbf{P2:}~~~  \max_{\Omega_I,\Omega_E,\mathbf p}~~&R_{AC}\\
  s.t.~~~~&\sum_{i\in\mathcal A}p_i\leq\sum_{i\in\Omega_E}\zeta Ph_i\\
  &p_i\geq0,~\forall i\in\mathcal A\\
  &\Omega_I\cap\Omega_E=\emptyset,~\Omega_I\cup\Omega_E=\mathcal A.
\end{align}
\end{subequations}

\textbf{P2} is a mixed integer programming (MIP) problem and usually NP-hard due to the combinatorial nature. In the following subsections, we discuss the optimal and suboptimal antenna clustering algorithms, respectively.

\subsection{Optimal Antenna Clustering}

In this subsection, we propose the optimal solution for the problem \textbf{P2}. Firstly, it is observed that  given any antenna partitions $\Omega_I$ and $\Omega_E$, the optimal power allocations follow \eqref{eqn:optp}, 
where $P_R=\sum_{i\in\Omega_E}\zeta Ph_i$.

Substituting the above results into \textbf{P2} to eliminate the power allocation variables $\mathbf p$, the problem is reduced to
%
a combinatorial optimization problem of set partition. The optimal antenna partition $\Omega_I^*$ and $\Omega_E^*$ can be found by exhaustively searching over all $2^{N}$ possible antenna combinations to maximize the resulting rate.

\subsection{Greedy Antenna Clustering}

The optimal antenna clustering algorithm by the exhaustive search is of exponentially increasing complexity
with the number of relay antennas $N$. In this subsection, we propose a low-complexity greedy algorithm for antenna clustering which is only with linear complexity $\mathcal O(N)$ instead of $\mathcal O(2^N)$ by the exhaustive search. The simulation results will show that the proposed greedy algorithm approaches to the optimal performance by the exhaustive search.

The key idea of the proposed greedy algorithm is switching an antenna into $\Omega_I$ or $\Omega_E$ based on the rate improvement. Specifically, at the beginning of the algorithm, we assume that $\Omega_I$ and $\Omega_E$ have only one antenna respectively. For the rest $N-2$ antennas, each antenna is put into one of $\Omega_I$ and $\Omega_E$ if the resulting rate improvement is greater than the other.

We present the greedy algorithm in Algorithm 2. It is obvious that the complexity of Algorithm 2 is $\mathcal O(N)$, which is linear in the number of the relay antennas.
 In \cite{Krikidis}, the algorithm selects $L$ antennas out of total $N$ antennas at the relay for information decoding and  the rest for energy harvesting, which is a ``binary knapsack problem". As binary knapsack problem is NP-complete, it needs a  high (even it is polynomial) computational complexity for finding an efficient solution. In \cite{ZZhou}, for every time the algorithm selects one best relay antenna switched from the energy harvesting set to the information decoding set. Thus, a total of $(N+1)N/2$ possibilities are needed to search at the worst case and the complexity of the greedy antenna clustering method in \cite{ZZhou} is $O(N^2)$. Therefore, our algorithm has much lower complexity compared with that of \cite{Krikidis} and \cite{ZZhou}.

\begin{algorithm}[!t]
\caption{Greedy Algorithm for \textbf{P2}}
\begin{algorithmic}[1]
\STATE \textbf{initialize} $\Omega_I=\Omega_E=\emptyset$.  Randomly select two antennas ($i=1,2$ without loss of generality) into $\Omega_I$ and $\Omega_E$, respectively.
\FOR{$i=3:N$}
\STATE Assume $\Omega_I=\Omega_I\cup i$, compute the end-to-end rate and denote it as $R_{AC}'$;
\STATE Assume $\Omega_E=\Omega_E\cup i$, compute the end-to-end rate and denote it as $R_{AC}''$; \IF {$R_{AC}'>R_{AC}''$} \STATE Update $\Omega_I=\Omega_I\cup i$; \ELSE \STATE Update $\Omega_E=\Omega_E\cup i$; \ENDIF
\ENDFOR

\end{algorithmic}
\end{algorithm}

\section{Simulations}

In this section, extensive numerical results are provided to evaluate the performance of the proposed algorithms. For the purpose of performance comparison, the TS scheme is also considered as a benchmark which can be solved similarly to the PS and the details are omitted here.


We consider a two-dimensional plane where the distance between the source and the destination is $10$, and the relay node is  located in a line between the source and the destination. The source-to-relay and relay-to-destination distances are denoted as $d$ and $10-d$, respectively, where $0<d<10$. Each channel fading is modeled as $c\cdot L^{-\theta}$, where $c$ is the Rayleigh fading factor, $L$ is  the distance, and $\theta$ is the path loss exponent which is set to be $2$. The number of relay antennas, $N=4$ and $N=8$, are both considered. The energy conversion efficiency coefficient is assumed as $\zeta=80\%$.

\subsection{Performance Comparison}

In this subsection, we consider the case where the relay node is at the middle of the source and the destination, i.e. $d=5$. We evaluate the proposed algorithms versus the source's transmit power $P$ in terms of SNR (dB).

\begin{figure}[t]
\begin{centering}
\includegraphics[scale=0.65]{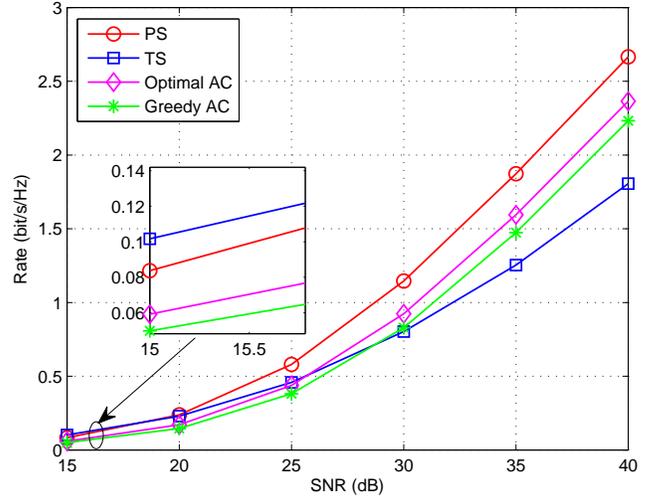}
\vspace{-0.1cm}
 \caption{Rate performance comparison of different algorithms versus the source's transmit power $P$ when $N=4$.}\label{fig:rate_N4}
\end{centering}
\vspace{-0.3cm}
\end{figure}
\begin{figure}[t]
\begin{centering}
\includegraphics[scale=0.65]{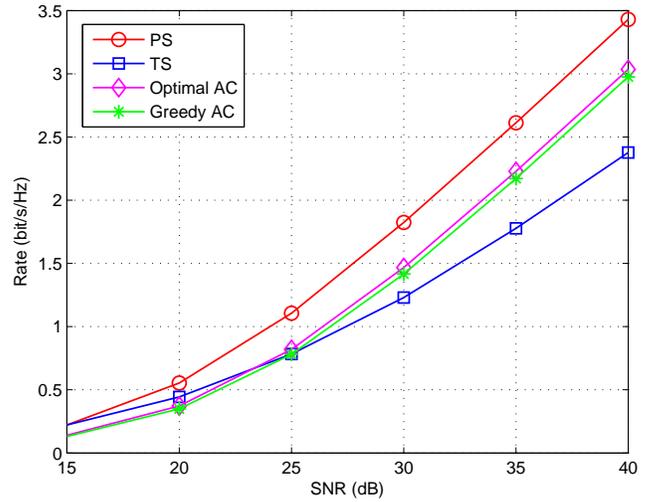}
\vspace{-0.1cm}
 \caption{Rate performance comparison of different algorithms source's transmit power $P$ when $N=8$.}\label{fig:rate_N8}
\end{centering}
\vspace{-0.3cm}
\end{figure}

Figs. \ref{fig:rate_N4} and \ref{fig:rate_N8}  compare the rate performance of different schemes with $N=4$ and $N=8$ relay antennas, respectively. From the two figures, we first observe that the achievable rates of all schemes are increasing with the source's transmit power $P$. This is because higher transmit power leads to more harvested power at the relay, so does for the system's SNR performance. Moreover, more relay antennas also result in higher rates. This is because that more relay antennas not only provide more spatial diversity but also harvest more power from the source.
We also observe that the TS scheme is slightly better than the PS scheme in low SNR region, while the PS scheme outperforms the TS scheme over a wide range of SNR, and the performance gain goes to large when SNR increases. Finally, it is shown that the greedy antenna clustering (AC) scheme performs closely to the optimal one, which demonstrates the effectiveness of the greedy AC scheme. It is also noted that when $N=8$,  Fig. \ref{fig:rate_N8}  shows that the greedy AC scheme almost approximates the optimal AC scheme. This means that more relay antennas improves the performance of the greedy AC scheme.

\begin{figure}[t]
\begin{centering}
\includegraphics[scale=0.65]{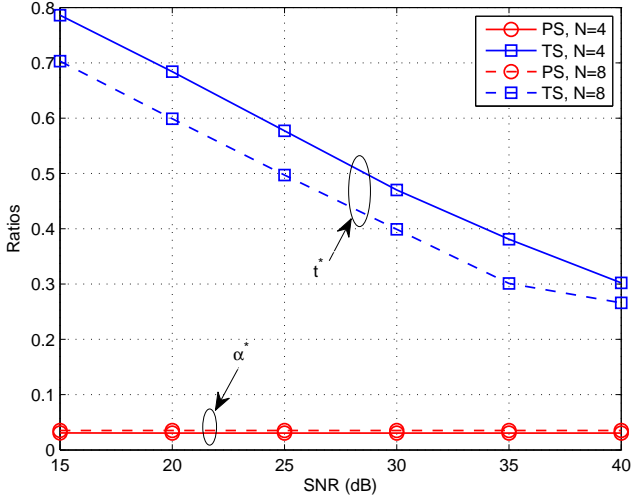}
\vspace{-0.1cm}
 \caption{Optimal $\alpha^*$ and $t^*$ versus source's transmit power $P$.}\label{fig:ratios}
\end{centering}
\vspace{-0.3cm}
\end{figure}

Then, we investigate the optimal power splitting ratio $\alpha^*$ and the optimal time allocation factor $t^*$ of the PS and TS respectively in Fig. \ref{fig:ratios}. It is observed that when the source's transmit power $P$ increases, the optimal power splitting ratio $\alpha^*$ of the PS scheme remains unchanged. In addition, only about $\alpha^*=0.04$ of the received power is used for information decoding, and the most received power is sent to energy harvesting. This shows that the wireless energy decay by path loss is the bottleneck of SWIPT. For the TS, it shows that the optimal time allocation factor $t^*$ decreases as the source's transmit power $P$ increases. This indicates that the wireless charging time can be reduced if the charging power becomes large.

\subsection{Impacts of Relay Locations}

In this subsection, we investigate the impacts of the relay locations where we vary the source-to-relay distance $d$ from $1$ to $9$. Here we fix the source's transmit power as $P=30$ dB.

\begin{figure}[t]
\begin{centering}
\includegraphics[scale=0.65]{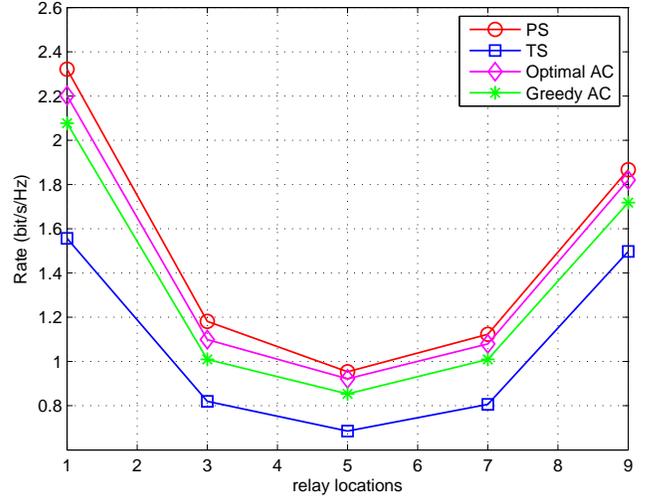}
\vspace{-0.1cm}
 \caption{Rate performance comparison of different algorithms versus the relay locations when $N=4$.}\label{fig:rate_N4_dis}
\end{centering}
\vspace{-0.3cm}
\end{figure}
\begin{figure}[t]
\begin{centering}
\includegraphics[scale=0.65]{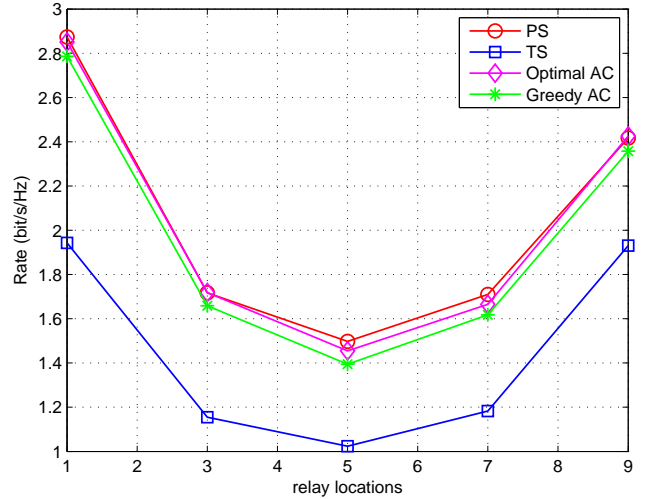}
\vspace{-0.1cm}
 \caption{Rate performance comparison of different algorithms versus the relay locations when $N=8$.}\label{fig:rate_N8_dis}
\end{centering}
\vspace{-0.3cm}
\end{figure}

Figs. \ref{fig:rate_N4_dis} and \ref{fig:rate_N8_dis} show the performance of different schemes versus the relay locations with $N=4$ and $N=8$ relay antennas, respectively. We observe that all schemes have the worst performance when the relay node is located at the middle of the source and the destination, and have  better performance when the relay node is close to the source or destination. The reason may be that, when the relay is close to the source, it yields a relatively higher energy harvesting efficiency. When the relay is close to the destination, a relatively better channel quality of the relay-to-destination link is available such that the system performance may be improved, although a relatively lower energy harvesting efficiency is achieved in this case. It is also observed that the PS scheme performs best and the TS scheme performs worst under this SNR. In addition, the performance of the greedy AC scheme is very close to that of the optimal AC scheme, especially the number of relay antennas is large (i.e., $N=8$). Moreover, we find that the optimal AC scheme almost achieves the same performance of the PS scheme when the number of relay antennas is large.

\begin{figure}[t]
\begin{centering}
\includegraphics[scale=0.65]{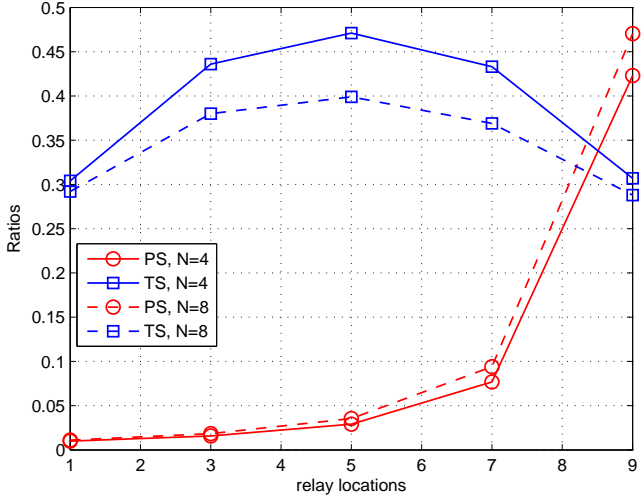}
\vspace{-0.1cm}
 \caption{Optimal $\alpha^*$ and $t^*$ versus the relay locations.}\label{fig:ratios_dis}
\end{centering}
\vspace{-0.3cm}
\end{figure}

Fig. \ref{fig:ratios_dis} shows the impacts of relay locations on the optimal power splitting ratio $\alpha^*$ and the optimal time allocation $t^*$. It first observes that the optimal power splitting ratio $\alpha^*$ increases when the relay node is moving to the destination. This means that more received power is split to information decoding in this case. It is also observed that the wireless charging time $t^*$ is firstly increased and then decreased when the relay moves away from the source. The possible reason is given in above paragraph.

\section{Conclusion}

In this paper, we studied SWIPT in multi-antenna DF relay networks, where the relay adopts the ``harvest-then-use" based energy harvesting strategy to forward information. The PS relay receiver architecture was first considered and the corresponding joint resource allocation problem was formulated and solved optimally. Then the antenna clustering scheme was proposed to ease implementation cost, where the relay antennas are partitioned into two disjoint groups with one for information decoding and the other for energy harvesting. Optimal and suboptimal clustering algorithms were developed.

A few important conclusions have been made through theoretical analysis and extensive simulations. Firstly, for the multi-antenna DF relaying, the harvested power at the relay is optimally allocated to the relay antennas based on the proportional criterion according to their second hop channel gains. Secondly, for the PS scheme, the optimal power splitting ratios of the relay antennas are the same. Thirdly, the PS scheme has a large performance gain over the TS scheme in high SNR, while the TS is slightly better in low SNR. Fourthly, for all schemes, placing the relay in the middle of the source and the destination results in the worst performance, and it is better to place the relay closer to the source or destination. Last but not least, the optimal power splitting ratio of PS remains unchanged with the transmit power and the channel gains of the first hop, while the optimal time allocation factor of TS is decreasing with the transmit power.

\appendices
\section{Derivation of \eqref{eqn:optp}}\label{app:p}
We can obtain the following optimality condition for $\mathbf p^*$ by applying KKT conditions:

\begin{eqnarray}\label{eqn:Lap}
  \frac{\partial L(\lambda,\mu,\boldsymbol\alpha,\mathbf p)}{\partial p_i^*}\begin{cases}<0,~p_i^*=0\\
  =0,~p_i^*>0\end{cases}~\forall i,
\end{eqnarray}
where
\begin{equation}
\frac{\partial L(\lambda,\mu,\boldsymbol\alpha,\mathbf p)}{\partial p_i^*}=\frac{(1-\lambda^*)\left(\sum_{j\in\mathcal A}\sqrt{p_j^*g_j}\right)\sqrt{g_i}}{\delta\left(1+\left(\sum_{j\in\mathcal A} \sqrt{p_j^*g_j}\right)^2\right)\sqrt{p_i^*}}-\mu^*.
\end{equation}

Thus, the optimality condition \eqref{eqn:Lap} becomes

\begin{equation}\label{eqn:D}
  \sqrt{\frac{g_i}{p_i^*}}\leq\frac{\delta\left(1+\left(\sqrt{\sum_{j\in\mathcal A}p_j^*g_j}\right)^2\right)\mu^*}{(1-\lambda^*)\left(\sum_{j\in\mathcal A}\sqrt{p_j^*g_j}\right)}\triangleq D,~\forall i,
\end{equation}
which shows that the the right side of the formula is the same for each $i$ and thus is a constant defined as $D$. This means that $\sqrt{g_i/p_i^*}\leq D$ should be satisfied for each $i$, and then implies that $p_i^*>0$ (or $p_i^*\neq0$) for each $i$ and the equality in \eqref{eqn:D} always holds. Thus, we have

\begin{equation}\label{eqn:p1}
p_i^*=\frac{g_i}{D^2},~\forall i.
\end{equation}

Moreover, according to the complementary slackness \cite{Boyd}, we have

\begin{equation}
\mu^*\left[\sum_{i\in\mathcal A}\zeta Ph_i(1-\alpha_i^*)-\sum_{i\in\mathcal A}p_i^*\right].
\end{equation}
As $\mu^*\neq0$ at the optimum, it must be

\begin{equation}\label{eqn:p2}
\sum_{i\in\mathcal A}p_i^*=\sum_{i\in\mathcal A}\zeta Ph_i(1-\alpha_i^*)\triangleq P_R,
\end{equation}
where $P_R$ is the total harvested power of the relay.

%
%
%
%
%

Combining \eqref{eqn:p1} and \eqref{eqn:p2}, the optimal relay power allocation \eqref{eqn:optp} is obtained.

\bibliographystyle{IEEEtran}
\bibliography{IEEEabrv,mimorelay}


\end{document}